\documentclass[12pt,english]{article}
    \usepackage[T1]{fontenc}
    \usepackage[latin9]{inputenc}
    \usepackage{geometry}
    \usepackage{amsthm}
    \usepackage{amsmath}
    \usepackage{amssymb}
    \usepackage{graphicx}
    \usepackage{chngcntr}
    \usepackage{apptools}
    \usepackage[authoryear]{natbib}
    \makeatletter
    \numberwithin{equation}{section}
    \numberwithin{figure}{section}
    
    \usepackage{amsfonts}
    \usepackage{amsmath}
    \usepackage{amsthm}
    \usepackage{color}
    \usepackage[doublespacing]{setspace}
    \usepackage[bottom]{footmisc}
    \usepackage{indentfirst}
    \usepackage{endnotes}
    \usepackage{graphicx}
    \usepackage{rotating}
    \usepackage{lscape}
    \usepackage{longtable}
    \usepackage{pdflscape}
    \usepackage{floatpag}
    \usepackage{array}
    \usepackage{calligra}
    \usepackage[unicode=true,pdfusetitle, bookmarks=true,bookmarksnumbered=false,bookmarksopen=false, breaklinks=false,pdfborder={0 0 1},colorlinks=false]{hyperref}
    \hypersetup{citecolor=blue, colorlinks=true}
    \usepackage[all]{xy}
    \usepackage{pifont}

    \usepackage[normalem]{ulem}
    \DeclareMathAlphabet{\mathcalligra}{T1}{calligra}{m}{n}
    \bibpunct{(}{)}{,}{a}{,}{,} 
        \theoremstyle{definition}

    \newtheorem{theorem}{Theorem}

    \newtheorem{conjecture}{Conjecture}
    \newtheorem{corollary}{Corollary}

    \newtheorem{definition}{Definition} %I think this should be numbered with thms
    \DeclareMathAlphabet{\mathcalligra}{T1}{calligra}{m}{n}
    \counterwithout{figure}{section}

    \date{}
    \makeatother
    
    \usepackage{babel}

    \long\def\ch#1{#1}            %changes are not indicated
    
    \long\def\info#1{}            %changes are not indicated

\begin{document}
% Title portion. Note the short title for running heads 
\title {Complexity of Stability in Trading Networks\thanks{
    We would like to thank Peter Bir\'{o}, John Hatfield, Ravi Jagadeesan, Scott Kominers, Mike Ostrovsky, and Alex Westkamp for fruitful conversations about this project. 
    Fleiner and Jank\'{o} were supported by the Hungarian Scientific Research Fund (OTKA grant K-108383) as well as the
    MTA-ELTE Egerv\'ary Research Group. Jank\'{o} also acknowledges the support of the MTA-BCE Strategic Interaction Research Group. Part
    of the research was carried out during two working visits at Keio
    University. Schlotter was also supported by the Hungarian Scientific Research Fund (OTKA grant K-108947). This work was supported by the Economic and Social Research Council grant number ES/R007470/1. Teytelboym is grateful for the hospitality of the Budapest University of Technology and Economics in August 2016. 
}}

\author{
Tam\'{a}s Fleiner\thanks{{Budapest University of
Technology and Economics. E-mail: {\tt fleiner@cs.bme.hu}}.}
%, Budapest
\and 
Zsuzsanna Jank\'o\thanks{{University of Hamburg. E-mail: {\tt
    zsuzsanna.janko@uni-hamburg.de}}.}
\and
Ildik\'{o} Schlotter\thanks{{Budapest University of Technology and Economics. E-mail: {\tt
ildi@cs.bme.hu}}.}
\and
Alexander Teytelboym\thanks{{University of Oxford. Email: {\tt alexander.teytelboym@economics.ox.ac.uk}}.}
}
\date{\today}
\maketitle
\singlespacing
% note that the abstract must come before \maketitle
\begin{abstract}
    Efficient computability is an important property of solution concepts in matching markets. We consider the computational complexity of finding and verifying various solution concepts in trading networks---multi-sided matching markets with bilateral contracts---under the assumption of full substitutability of agents' preferences. It is known that outcomes that satisfy \textit{trail stability} always exist and can be found in linear time. Here we consider a slightly stronger solution concept in which agents can simultaneously offer an upstream and a downstream contract. We show that deciding the existence of outcomes satisfying this solution concept is an $\mathsf{NP}$-complete problem even in a special (flow network) case of our model. It follows that the existence of \textit{stable} outcomes---immune to deviations by arbitrary sets of agents---is also an $\mathsf{NP}$-hard problem in trading networks (and in flow networks). Finally, we show that even verifying whether a given outcome is stable is $\mathsf{NP}$-complete in trading networks.
    \end{abstract}
    % note: this command has been disabled to remove the ACM copyright block. Sorry...
    
    % Renew this after \maketitle if the default list of authors is too long for headers
    %\renewcommand{\shortauthors}{W.\ Vickrey et.\ al.}
    \noindent \begin{flushleft}
\par\end{flushleft}{\small \par}

\noindent \begin{flushleft}
\par\end{flushleft}{\small \par}

    \onehalfspacing
    \section{Introduction}

    One of the most important features of the marriage market or the college admissions problems is that (pairwise) stable outcomes can always be found efficiently using the celebrated Deferred Acceptance algorithm \citep{GaSh:62}. Efficient computation is key for practical applications of the Deferred Acceptance algorithm and its variants \citep{Roth:1984a,AbSo:03,SoSw:13,HaKo:14}. Moreover, since in standard (one-to-one) marriage markets or (many-to-one) college admissions problems (pairwise) stable outcomes coincide with the core, there is no obvious need to select among various solution concepts.
    
    However, in more complex, many-to-many matching markets there are many different, yet economically natural and interpretable, solution concepts \citep{Blai:88,Soto:99,EcOv:06,KlWa:09}. For example, pairwise stable outcomes do not coincide with the core \citep{Blai:88}. What makes a good solution concept? At the very least, it should make falsifiable predictions. But in economics good solution concepts often derive their appeal for normative reasons: they have sensible properties and make intuitive sense in particular settings. Efficient computation can serve as one such desirable property. As \citet[pp. 29-30]{Papa:07} argues:
    \begin{quote}
    The reason is simple: If an equilibrium concept is not efficiently computable, much of its credibility as a prediction of the behavior of rational agents is lost -- after all, there is no clear reason why a group of agents cannot be simulated by a machine. Efficient computability is an important modelling prerequisite for solution concepts.
    \end{quote} 
If agents cannot efficiently find certain deviations from an outcome in a matching market, then a stability concept that is based on these deviations may be too strong. On the other hand, if it is computationally hard to verify whether an outcome satisfying a particular solution concept exists then this limits the applicability of this solution concept for matching market design.

    In this paper, we consider the computational complexity of various solution concepts in \emph{trading networks}. A trading network is a matching market in which heterogeneous firms can sign many contracts that specify the terms of the trade, quality of the product, price etc.\ with their suppliers (upstream) and with their buyers (downstream). Following a seminal contribution by \citet{Ostr:08}, we assume that agents' preferences satisfy \emph{full substitutability}, that is, upstream and downstream contracts are complements, but contracts on the same side are substitutes. Our model can capture not only wealth effects, but also distortionary frictions, such as sales taxes or bargaining costs \citep*{FlJaJaTe:17}. Finding a solution concept with attractive computational as well as economic properties would allow trading networks to be deployed in a variety of empirical \citep{fox2017specifying} and practical applications \citep{morstyn2018bilateral}. 
    
    In our model, \textit{stable} outcomes, i.e., those immune to arbitrary blocks by sets of firms, do not necessarily exist~\citep{HaKo:12}.\footnote{Stable outcomes exist in supply chains in which no firm can be simultaneously upstream and downstream from another \citep{Ostr:08,West:10,HaKo:12}. Stable outcomes also exist (and coincide with competitive equilibria) in a variant of our model with continuous prices and no distortionary frictions \citep{HaKoNiOsWe:11,FlJaJaTe:17}.} 
Following \citet*{FlJaTaTe:16}, we first point out that outcomes satisfying \textit{trail stability}, an extension of pairwise or chain stability \citep{Ostr:08}, always exist and can be found in linear time (in the number of contracts) by an extension of the Deferred Acceptance algorithm (Theorem \ref{main}). Trail-stable outcomes are immune to locally blocking trails (sequences of distinct contracts in which a buyer of one contract is the seller in the next). In a locally blocking trail, agents simultaneously accept pairs of upstream and downstream contracts along the trail but there must also be agents that ``kick off'' (and ``complete'') the blocking trail by unilaterally offering (or ``accepting'') a single downstream (or a single upstream) contract.
Trail stability is an attractive solution concept: in a variant of our model with continuous prices and distortionary frictions, trail-stable outcomes essentially coincide with competitive equilibrium outcomes \citep*{FlJaJaTe:17}.

    We then relax the requirement of locally blocking trails that agents can only make an upstream (downstream) offer following the receipt of a downstream (upstream) offer. Thus we allow agents to offer an upstream and a downstream contract \textit{simultaneously}, forming \textit{blocking cycles}. As a result, we consider the existence of \textit{path-or-cycle-stable} outcomes which are immune to blocking sets in the form of paths (trails in which every agent is distinct) or cycles.\footnote{We do not suggest that path-or-cycle stability has an intuitive economic interpretation. We use this solution concept to give the strongest computational complexity result possible and to reveal the structure of stable outcomes in flow networks.} Path-or-cycle-stable outcomes do not always exist in trading networks. We show that the decision problem of whether a path-or-cycle-stable outcome exists is $\mathsf{NP}$-complete even in a \textit{flow network} \citep{Flei:14}, which is a special case of a trading network (Theorem~\ref{thm:path+cycle}).\footnote{\citet{garey1979computers} provide definitions of computational complexity classes used in this paper.} The proof reduces the problem to the problem of partitioning a directed graph into two acyclic sets, which is known to be $\mathsf{NP}$-complete \citep{Bokal:02}. Since trading networks are more general than flow networks, it follows that the problem of deciding whether a path-or-cycle-stable outcome exists in a trading network is also $\mathsf{NP}$-hard. Therefore, the ability to coordinate upstream and downstream contract offers minimally renders the computational problem for existence intractable. Our result superficially resembles the problems of determining the existence of pairwise stable outcomes in two-sided hospital-resident markets with couples \citep{mcdermid2010keeping}, with sizes \citep{delacretaz2014matching}, or with multidimensional constraints \citep{delacretaz2016refugee}: in these models, as in ours, finding stable outcomes is hard due to the presence of various constraints. However, as far we know, none of our results has appeared elsewhere in the literature.
    
    Finally, we turn our attention to stable outcomes analyzed in the context of trading networks by \citet{HaKoNiOsWe:14}. While stable outcomes do not always exist in trading networks, \citet{HaKoNiOsWe:14} show that in trading networks blocking sets can be decomposed into certain blocking trails (under a monotonicity condition).
    We first show that in flow networks stable outcomes coincide with path-or-cycle-stable outcomes (Theorem \ref{thm:set-stable-flow}). It follows that deciding whether a stable outcome exists in a flow network 
    %or in a trading network 
is $\mathsf{NP}$-complete (Corollary \ref{cor:set-stable-flow}). Finally, we show that even \textit{verifying} that a particular outcome is (not) stable is $\mathsf{NP}$-complete (Theorem \ref{grouphopeless}). The proof provides a reduction from the set partition problem, which is known to be (weakly) $\mathsf{NP}$-complete. 
    
    Section \ref{sec.modelsec} introduces the full model of trading networks and the special case of flow networks (\ref{sec:flow}). In Section \ref{sec.stabconcepts}, we introduce trail stability (\ref{sec.trailstab}), path-or-cycle stability (\ref{sec.pathcyclestab}), and stability (\ref{sec.setstab}), and state the main results about the computational complexity of finding outcomes that satisfy these solution concepts. Section \ref{sec.open} leaves an open question concerning another solution concept---\emph{weak trail stability}---introduced by \citet*{FlJaTaTe:16}. Section \ref{sec.conclusion} is a conclusion.
    
    % quote

    \section{Model}\label{sec.modelsec}
    Our notation follows \citet*{FlJaTaTe:16}.
    \subsection{Ingredients}
    In a \emph{trading network}, there is finite set of agents (firms or consumers) $F$ and a finite set of contracts $X$.
     A contract $x\in X$ is a bilateral agreement between a buyer $b(x)\in F$
    and a seller $s(x)\in F$. Hence, $F(x):=\{s(x),b(x)\}$ is the
    set of firms associated with contract $x$ and, more generally, $F(Y)$
    is the set of firms associated with contract set $Y\subseteq X$.
    Call $X^B_{f}:=\{x\in X \mid b(x)=f\}$ and $X^S_{f}:=\{x\in X \mid s(x)=f\}$
    the sets of $f$'s upstream and downstream contracts -- for which
    $f$ is a buyer and a seller, respectively. Clearly, $Y^B_{f}$ and
    $Y^S_{f}$ form a partition over the set of contracts $Y{}_{f}:=\{y\in
    Y \mid f\in F(y)\}$ 
    which involve $f$, since an agent cannot be a buyer and a seller
    in the same contract. A firm $f$ is a \textit{terminal seller} if there are no upstream contracts for $f$ in the network and $f$ is a \textit{terminal buyer} if the network does not contain any downstream contracts for $f$. An agent who is either a terminal buyer or a terminal seller is called a \textit{terminal agent}. In an \emph{acyclic} trading network, no agent simultaneously buys and sells from another agent, even via intermediaries.
    
    Every firm has a choice function $C^{f}$, such that $C^{f}(Y_{f})\subseteq Y_{f}$
    for any $Y_{f}\subseteq X_{f}$.\footnote{Since firms only care about their own contracts, we can write $C^{f}(Y)$ to mean $C^{f}(Y_{f})$.} We say that choice functions of $f\in F$ satisfy the \textit{irrelevance of rejected contracts (IRC)} condition if for any $Y\subseteq X$ and $C^{f}(Y)\subseteq Z\subseteq Y$ we have that $C^{f}(Z)=C^{f}(Y)$ \citep{Blai:88,Alka:02,Flei:03,Eche:07,AySo:12}.
    
    For any $Y\subseteq X$ and $Z\subseteq X$, define the \textit{chosen}
    set of upstream contracts 
    \begin{equation}
    C_{B}^{f}(Y|Z):= C^{f}(Y^B_f\cup Z^S_f)\cap X^B_f
    \end{equation}
    which is the set of contracts $f$ chooses as a buyer when $f$ has
    access to upstream contracts $Y$ and downstream contracts $Z.$ Analogously,
    define the chosen set of downstream contracts 
    \begin{equation}
    C_{S}^{f}(Z|Y):= C^{f}(Z^S_f \cup Y^B_f)\cap X^S_f
    \end{equation}
    For brevity, we will also use $C^f(Y|Z):=(C_B^f(Y|Z)|C_S^f(Z|Y))$, so $C^f(Y|Z)=(Y'|Z')$ means that if $f$ is offered upstream and downstream contracts $Y$ and $Z$, respectively, then $Y'$ and $Z'$ are those among them that $f$ chooses (with $Y' \subseteq Y$ and $Z' \subseteq Z$).
    We also define the \textit{rejected} sets of contracts $R_{B}^{f}(Y|Z):= Y_f\setminus C_{B}^{f}(Y|Z)$
    and $R_{S}^{f}(Z|Y):= Z_f\setminus C_{S}^{f}(Z|Y)$. An \textit{outcome}
    $A\subseteq X$ is a set of contracts.
    
    A set of contracts $A\subseteq X$ is \textit{individually rational} for an agent $f\in F$ if $C^{f}(A_{f})=A_{f}$. We call set $A$ \textit{acceptable} if $A$ is 
    individually rational for all agents $f\in F$. For sets of
    contracts $W,A\subseteq X$, we say that $A$ is \textit{$(W,f)$-acceptable} if
    $A_{f}\subseteq C^f(W_{f}\cup A_{f})$, i.e., if the agent $f$ chooses all
    contracts from set $A_f$ whenever she is offered $A$ alongside
    $W$. Set of contracts $A$ is \textit{$W$-acceptable} if $A$ is $(W,f)$-acceptable
      for all agents $f\in F$. Note that contract set $A$ is individually rational for agent $f$ if and only if it is
      $(\emptyset,f)$-acceptable.
      
    \subsection{Assumptions on choice functions}
    
    We can now state our key assumption on choice functions introduced by \citet{Ostr:08}.
    \begin{definition}
    \label{def:full-sat}
    
    Choice functions of $f\in F$ are \textit{fully
    substitutable} if for all $Y'\subseteq Y\subseteq X$ and $Z'\subseteq Z\subseteq X$
    they are:
    \begin{enumerate}
    \begin{minipage}{.5\textwidth}
    \item \textit{Same-side substitutable} (SSS):
    
    \begin{enumerate}
    \item $R_{B}^{f}(Y'|Z)\subseteq R_{B}^{f}(Y|Z)$
    \item $R_{S}^{f}(Z'|Y)\subseteq R_{S}^{f}(Z|Y)$
    \end{enumerate}
    \end{minipage}
    \begin{minipage}{.5\textwidth}
    \item \textit{Cross-side complementary} (CSC):
    
    \begin{enumerate}
    \item $R_{B}^{f}(Y|Z)\subseteq R_{B}^{f}(Y|Z')$
    \item $R_{S}^{f}(Z|Y)\subseteq R_{S}^{f}(Z|Y')$
    \end{enumerate}
    \end{minipage}
    \end{enumerate}
    \end{definition}
    
    Contracts are stitutable if every firm regards any of its
    upstream or any of its downstream contracts as substitutes, but its
    upstream and downstream contracts as complements. Hence, rejected
    downstream (upstream) contracts continue to be rejected whenever the
    set of offered downstream (upstream) contracts expands or whenever
    the set of offered upstream (downstream) contracts shrinks.

    \subsection{Special case: flow networks and flow-based choice functions}\label{sec:flow}
    We first define flow-based choice functions.
    \begin{definition} $C^f$ is \emph{flow-based} if 
    \begin{itemize}
    \item whenever $f \in F$ is a terminal agent, $C^f(Y_f)=Y_f$ for any $Y \subseteq X$;
    \item whenever $f \in F$ is not a terminal agent
    \begin{itemize}
    \item $f$ has linear-order preferences $\succ^B_f$ over $X_f^B$ and linear-order preferences $\succ^S_f$ over $ X_f^S$, and
    \item given access to upstream contracts $Y \subseteq X_f^B$ and downstream contracts $Z \subseteq X_f^S$, a firm $f$ with flow-based choice function will choose its $k$ most preferred upstream contracts from $Y$ according to $\succ^B_f$ and $k$ most preferred downstream contracts from $Z$ according to $\succ^S_f$ where $k:=\min\{|Y|,|Z|\}$. 
     \end{itemize}
     \end{itemize}
     \end{definition}
Since non-terminal agents with flow-based choice functions always pick the same number of upstream and downstream contracts, their preferences satisfy the so-called \emph{Kirchhoff} equality.\footnote{Flow-based choice functions are a special case of ``separable'' choice functions \citep*{FlJaTaTe:16}.} It is immediate that flow-based choice functions defined in this way are fully substitutable and satisfy IRC. 

A \emph{flow network} is a trading network in which there are exactly two terminal agents (and the remaining firms are non-terminal) and all choice functions are flow-based \citep{Flei:14}.\footnote{In fact, what follows is a simplification of Fleiner's flow network model where all capacities have value~1 and flow values must be integral.}

    \section{Solution concepts and their computational complexity}\label{sec.stabconcepts}
    
    \subsection{Trail stability}\label{sec.trailstab}
    In order to explain our first solution concept, we first define a \textit{trail}.
\begin{definition}
A non-empty set of contracts $T$ is a
\textit{trail} if its elements can be arranged in some order $(x_{1},\ldots,x_{M})$ such that
$b(x_{m})=s(x_{m+1})$ holds for all $m\in \{1,\ldots,M-1\}$ where $M=|T|$.
\end{definition}
    
    The following solution concept was introduced by \citet*{FlJaTaTe:16}.
    
\begin{definition} An outcome $A\subseteq X$ is \textit{trail-stable} if 
\begin{enumerate}
\item $A$ is acceptable.
\item There is no trail $T=\{x_1,x_2,\ldots,x_M\}$, such that $T\cap
  A=\emptyset$  and
\begin{enumerate}
\item$\{x_{1}\}$ is $(A,f_1)$-acceptable for $f_{1}=s(x_{1})$, and
\item$\{x_{m-1}, x_{m}\}$ is $(A,f_m)$-acceptable for
$f_{m}=b(x_{m-1})=s(x_{m})$ whenever $1<m\le M$ and
\item$\{x_{M}\}$ is $(A,f_{M+1})$-acceptable for $f_{M+1}=b(x_{M})$.
\end{enumerate}
The above trail $T$ is called a \textit{locally blocking trail to
  $A$}.
\end{enumerate}

\end{definition}
    
    Trail stability is a natural solution concept when firms interact
    mainly with their buyers and suppliers and deviations by arbitrary
    sets of firms are difficult to arrange. In a trail-stable outcome,
    no agent wants to drop his contracts and there exists no sequence
     of \textit{consecutive}
    bilateral contracts comprising a trail such that any intermediate agent who is offered a downstream (upstream) contract along the trail wants to choose it alongside the subsequent upstream (downstream) contract in the trail. Importantly, we require that the first (final) agent wants to unilaterally offer  (accept) the first (final) contract in the trail.\footnote{The trail and the order of conditional acceptances can, of course,
    be reversed with $f_{M+1}$ offering the first upstream contract to
    seller $f_{M}$ and so on.} For trail stability, we do not require that intermediate agents accept all the contracts along the locally blocking trail. Instead they are simply required to myopically accept pairs of upstream and downstream contracts as they appear along the trail. We strengthen this requirement in Section \ref{sec.open}.
    
    \citet*{FlJaTaTe:16} proved that under full substitutability trail-stable outcomes always exist in trading networks and, under certain conditions, have a familiar lattice structure. Trail-stable outcomes may not be Pareto-efficient. %We reproduce the existence proof in the Appendix. 
    In a variant of our model with continuous prices, \citet*{FlJaJaTe:17} showed that whenever there are distortionary frictions, such as sales taxes, in the economy, trail-stable outcomes essentially coincide with competitive equilibrium outcomes.\footnote{The First Fundamental Welfare Theorem can, of course, fail in the presence of distortionary frictions.} Here we focus on the computational properties of trail-stable outcomes.
    
    \begin{theorem}[\citealp*{FlJaTaTe:16}]\label{main}
    Suppose that in a trading network choice functions satisfy full substitutability and IRC. Then a trail-stable outcome exists and can be found in time linear in the number of contracts. 
    \end{theorem}
    The proof of Theorem \ref{main} follows immediately from the proof of existence of trail-stable outcomes provided by \citet*{FlJaTaTe:16} and \citet{adachi2017stable}. Trail-stable outcomes are found by a generalized Gale-Shapley algorithm defined in \citet*{FlJaTaTe:16}. In the generalized Gale-Shapley algorithm, at least one contract is rejected at each step and since the number of contracts is finite, the number of steps required to find a trail-stable outcome is bounded by the number of contracts.
    
    \subsection{Path-or-cycle stability}
    \label{sec.pathcyclestab}
    The definition of trail stability requires that there be initial and final agents that would unconditionally offer or accept an upstream or a downstream contract while all the intermediaries only make a downstream (upstream) offer after receiving an upstream (downstream) offer. Let us now relax this condition and allow agents to form blocking cycles: now every agent can offer an upstream at a downstream contract simultaneously without having to accept them individually. This is a mild strengthening of trail stability (whenever there is at most one contract between agents) that treats the initial and the terminal agents in the blocking trail in the same way as the intermediate agents in the blocking trail and allows for minimal additional coordination among blocking agents.
    
    For the sake of further simplicity, let us only focus on \textit{paths}, i.e., trails in which \textit{all agents are distinct}.
    \begin{definition}A non-empty set of contracts $P$ ($|P|=M$) is a
\textit{path} if its elements can be arranged in some order $(x_{1},\ldots,x_{M})$ such that
\begin{itemize}

\item $b(x_{m})=s(x_{m+1})$ holds for all $m\in \{1,\ldots,M-1\}$, and
\item all firms $s(x_1), s(x_2), \dots, s(x_M)$, and $b(x_M)$ are distinct.
\end{itemize}

    \end{definition}
    Now, a path that ``returns'' to its origin agent is a \textit{cycle}.
    \begin{definition}
    A non-empty set of contracts $C$ ($|C|=M$) is a
\textit{cycle} if its elements can be arranged in some order $(x_{1},\ldots,x_{M})$ such that
\begin{itemize}

\item $b(x_{m})=s(x_{m+1})$ holds for all $m\in \{1,\ldots,M-1\}$, 
\item all firms $s(x_1), s(x_2), \dots, s(x_M)$, and $b(x_M)$ are distinct, and 
\item $b(x_M)=s(x_1)$.
\end{itemize}

    \end{definition}
    Cycles are therefore also trails.
    \begin{definition} An outcome $A\subseteq X$ is \textit{path-or-cycle-stable} if 
    \begin{enumerate}
    \item $A$ is acceptable.
    \item There is no path or cycle $B$ such that $B \cap A = \emptyset$ and $B$ is $(A,f)$-acceptable for each $f \in F(B)$.\\ 
    Such paths or cycles are called \emph{blocking paths} and \emph{blocking cycles}.\footnote{Note that path-or-cycle stability is weaker than chain stability in the sense of \citet{HaKoNiOsWe:14}.} 
    \end{enumerate}
    \end{definition}
    
    %[NOTE: We can move the following two paragraphs into the proof.]
    An interesting property of flow-based choice functions is that given an outcome $A \subseteq X$, any cycle $C$ disjoint from $A$ is a blocking cycle, as any firm which is offered a pair of additional upstream and downstream contracts will accept them. We will use this property to prove our second key result.
    
    \begin{theorem}
    \label{thm:path+cycle}
    %Suppose that in a flow network choice functions are flow-based. Then 
    It is $\mathsf{NP}$-complete to decide if a flow network admits a path-or-cycle-stable outcome. 
    \end{theorem}
    
    \begin{proof}
    The problem is in $\mathsf{NP}$, since given an outcome $A$, we can check in linear time (with respect to the number of contracts) whether it admits a blocking path or a cycle: we only have to check that the contracts $X \setminus A$ do not contain a cycle or a path starting with an \ch{($A,\cdot)$-acceptable} contract and ending with an \ch{($A,\cdot)$-acceptable} contract; both these tasks can be decided using, e.g., some variant of the depth-first search (DFS) algorithm on the directed graph representing $X \setminus A$.
    %taking the directed  graph spanned by all contracts not in $A$, we only have to check that it is acyclic and does not contain a path starting with an seller-active contract and ending with a buyer-active contract; both these tasks can be decided using, e.g., some variant of the depth-first search (DFS) algorithm.
    
    To prove $\mathsf{NP}$-hardness, we present a polynomial reduction from the following problem: given a directed graph $D$, decide whether it is possible to partition the vertices of $D$ into two acyclic sets $V_1$ and $V_2$. Here, a set $U$ of vertices is \emph{acyclic}, if there is no directed cycle in $D[U]$. This problem was proved to be $\mathsf{NP}$-complete by~\citet{Bokal:02}.

    \begin{figure}[t]
    \begin{center}
    \includegraphics[scale=1.2]{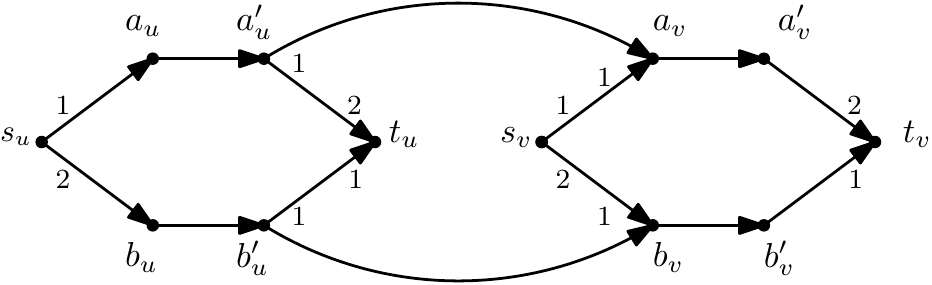}
    \caption{Illustration of two node gadgets, $D_u$ and $D_v$, and the two contracts in $Z$ connecting them; the figure assumes $uv \in E$. The number on some contract $x$ written next to some firm indicates the rank of the contract $x$ in the preference ordering of the given firm.}
    \label{fig-reduction}
    \end{center}
    \end{figure}
    
    Given our input $D=(V,E)$, we construct a set $X$ of contracts and a set $F$ of firms with flow-based choice functions
    (see Fig.~\ref{fig-reduction} for an illustration). 
    There will be at most one upstream and at most one downstream contracts between any two firms, so we will denote a contract $x$ with $s(x)=f_1$ and $b(x)=f_2$ as $f_1 f_2$. 
    First, we introduce a \textit{node gadget} $D_v$ for each $v \in V$; 
    the set of firms in $D_v$ is $\{s_v,a_v,a'_v,b_v,b'_v,t_v\}$ and 
    the set of contracts in $D_V$ is $\{s_v a_v, s_v b_v, a_v a'_v, b_v b'_v, a'_v t_v, b'_v t_v\}$. 
    Next, we add terminal firms $s$ and $t$, 
    together with the contracts $s s_v$ and $t_v t$ for each $v \in V$. 
    Finally, we add the contract set $Z= \{ a'_u a_v, b'_u b_v \mid uv \in E\}$. 
    
    Instead of describing the full preferences of each firm over its upstream and downstream contracts, we only define a partial ordering (and assume that the preferences of each firm respect this partial order). Namely, for each $v \in V$, we let $s_v$ prefer $s_v a_v$ to $s_v b_v$, and we let $t_v$ prefer $b'_v t_v$ to $a'_v t_v$.  \ch{In Fig.~\ref{fig-reduction}, we indicate the rankings of the contracts with numbers 1 (highest rank), 2 (second highest) etc.} Additionally, we let any firm in $\{a_v,a'_v,b_v,b'_v\}$ prefer all contracts not in $Z$ to contracts in $Z$. 
    
    We claim that there exists an outcome in $X$ admitting neither blocking paths nor blocking cycles if and only if the vertices of $D$ can be partitioned into two acyclic sets.
    
    ``$\Rightarrow$'': Let us suppose that there exists an outcome $A \subseteq X$ that does not admit any blocking paths or cycles. 
    
    We show that $Z \cap A = \emptyset$. To see this, first consider a contract $a'_u a_v \in Z$, and suppose for contradiction that $a'_u a_v \in A$. Since $a_v$ has only one downstream contract $a_v\ch{a'_v}$, this means that the contract $s_v a_v$ cannot be contained in $A$ (because of the Kirchhoff equality for $a_v$). Note also that $s_v a_v $ is \ch{$(A, a_v)$-acceptable}, because it is a contract preferred by $a_v$ to $a'_u a_v$. Consider the path $P$ from $s$ through $s_v$ to $a_v$. Clearly, if neither contract on $P$ is in $A$, then it is a blocking path, otherwise the contract $s_v a_v$ is \ch{$(A,s_v)$-acceptable} and hence a blocking path itself, a contradiction. 
    So suppose now $b'_u b_v \in A$. Arguing analogously as before, we can prove that either the path from $b'_u$ through $t_u$ to $t$, or simply the contract $b'_u t_u$ is blocking. Thus we obtain that $A$ cannot contain any contracts from $Z$, and only contains contracts within node gadgets and contracts where the seller or the buyer is the terminal agent $s$ or the terminal agent $t$, respectively. 
    
    Therefore, we know that for each $v \in V$, at most one of the contracts $a_v a'_v$ and $b_v b'_v$ can be contained in $A$ (since $s_v$ can choose at most one of its downstream contracts by the Kirchhoff equality). 
    Let $Q=\{v \in V \mid a_v a'_v \notin A\}$, and let $R=V \setminus Q$; clearly $b_v b'_v \notin A$ for any $v \in R$. It is not hard to see that both $Q$ and $R$ are acyclic. Indeed, any cycle within vertices of $Q$ in $D$ corresponds to a cycle using only contracts of $Z$ and the contracts $a_v a'_v$, $v \in Q$, and such a cycle cannot exists as it would block $A$. The same argument works to show the acyclicity of $R$, proving the first direction of our reduction.

    ``$\Leftarrow$'': Assume now that $Q$ and $R$ are two acyclic subsets of $V$ forming a partition. We define an outcome $A \subseteq X$ that contains the contracts $s s_v$, $s_v b_v$, $b_v b'_v$, $b'_v t_v$, and $t_v t$ for each $v \in Q$, and similarly, the contracts $s s_v$, $s_v a_v$, $a_v a'_v$, $a'_v t_v$, and $t_v t$ for each $v \in R$. We claim that there is no blocking path or cycle for $A$. 
    
    To see this, observe that by the Kirchhoff equality a contract that is \ch{$(A,\cdot)$-acceptable} in itself for some firm (but is not contained in $A$) must be either $s_u a_u$ for some $u \in Q$ or $b'_v t_v$ for some $v \in R$. Since there is no path starting with a contract $s_u a_u$ and ending with a contract $b'_v t_v$, we know that no path can block $A$. 
    
    To show that $A$ admits no blocking cycles, we simply use that blocking cycles must be disjoint from $A$. Note that any cycle has to use at least one contract of $Z$, as node gadgets are acyclic. We partition the contracts in $Z$ into four sets as follows; recall that $A \cap Z=\emptyset$.
    Let $Z_{Q,R}$ denote those contracts in $Z$ that leave a node gadget $D_u$ with $u \in Q$ and arrive at a node gadget $D_v$ with $v \in R$; we define $Z_{R,Q}$, $Z_{Q,Q}$, and $Z_{R,R}$ analogously. To see that no contract $xy \in Z_{Q,R} \cup Z_{R,Q}$ can be part of a blocking cycle for $A$, note that $A$ contains either the unique upstream contract of $x$ or the unique downstream contract of $y$, by the definition of $A$. 
    In either case, any cycle that contains the contract $xy$ must also contain a contract in $A$, and hence cannot be a blocking cycle.
    
    By the same reasoning, no contract $b'_u b_v$ for some $u,v \in Q$ can be contained in a blocking cycle, since both the unique upstream contract of $b'_u$ and the unique downstream contract of $b_v$ are contained in $A$. Therefore, any contract of $Z_{Q,Q}$ used by a blocking cycle must be of the form $a'_u a_v$ for some $u,v \in Q$. Similarly, any contract of $Z_{R,R}$ used by a blocking cycle must be of the form $b'_u b_v$ for some $u,v \in R$. 
    By the structure of the network, this implies that no cycle can use contracts both from $Z_{Q,Q}$ and from $Z_{R,R}$. However, any blocking cycle that uses only contracts in $Z_{Q,Q}$ and contracts of the form $a_u a'_u$ with $u \in Q$ directly corresponds to a cycle within $D[Q]$. Similarly, any blocking cycle using only contracts in $Z_{R,R}$ and contracts of the form $b_u b'_u$ with $u \in R$ yields a cycle within $D[R]$. Therefore, the acyclicity of $Q$ and $R$ ensures that $A$ admits no blocking cycle, proving the correctness of our reduction.
    \end{proof}
    
    If we translate Theorem~\ref{thm:path+cycle} into the language of \emph{stable flows} introduced by~\citet{Flei:14}, we obtain that it is $\mathsf{NP}$-complete to decide whether a \emph{completely stable flow} exists in a given network with preferences, where a flow is completely stable if it admits neither blocking paths nor blocking cycles. 
    In fact, our statement holds not only in the discrete case (as implied directly by Theorem~\ref{thm:path+cycle}), but also in the continuous case where the flow can take real values as well; adjusting the proof of Theorem~\ref{thm:path+cycle} to this case is straightforward. Hence we settle a conjecture posed by~\citet{Flei:14}.
    
    \subsection{Stability}
    \label{sec.setstab}
    We now relax the assumption the blocking sets must be paths or cycles and consider general set blocks in trading networks \citep{HaKoNiOsWe:14} and flow networks \citep{Flei:14}.
    \begin{definition} An outcome $A\subseteq X$ is \textit{stable}  if:%\footnote{\citet{KlWa:09} call stable outcomes ``weak setwise stable'' and \citet{HaKo:12} call them ``stable'': we take the middle ground. \citet{West:10} applies the label ``group stable'' to ``setwise stable outcomes'' \citep{Soto:99,EcOv:06,KlWa:09}.}% 
    
    \begin{enumerate}
    \item $A$ is acceptable.
    \item There exist no non-empty set of contracts $Z\subseteq X$, such that $Z\cap A=\emptyset$ and $Z$ is $(A,f)$-acceptable for all $f\in F(Z)$; such sets are called blocking. 
    \end{enumerate}
    \end{definition}
    
    Stable outcomes are immune to deviations by arbitrary groups of firms, which can re-contract freely among themselves while keeping any of their existing contracts. Stable outcomes always exist in acyclic networks if choice functions are fully substitutable. However, \citet{Flei:14} and
    \citet{HaKo:12} showed that stable outcomes may not exist in
    general trading networks. 
    
    \subsubsection{Stability in flow networks}
    \label{sec.flowstab}
    We next prove that in flow networks path-or-cycle-stable outcomes coincide with stable outcomes. 
    \begin{theorem}
    \label{thm:set-stable-flow}
    In a flow network an outcome is path-or-cycle-stable if and only if it is stable.
    \end{theorem}
    
    \begin{proof}
    Let $X$ be a set of firms in a flow network. 
    Using the definitions, it is immediate that a stable outcome $A$ is also path-or-cycle-stable, as a blocking path or cycle is naturally a blocking set as well. 
    
    For the opposite direction, assume that $A$ is a path-or-cycle-stable outcome. 
    Towards contradiction, let us also assume that $B \subseteq X$ is a blocking set for $A$. 
    Suppose first that $B$ contains a cycle $B^C$. Then $B^C$ is disjoint from $A$ because $B^C \subseteq B \subseteq X \setminus A$. Moreover, $B^C$ is $(A,f)$-acceptable for any firm~$f$: if $f$ is a terminal then it accepts all contracts offered, and if $f$ is non-terminal then, by the Kirchhoff equality, it accepts $B^C$ alongside with $A$ since $B^C_f$ is either empty or it contains an upstream and a downstream contract for $f$. Hence, $B^C$ is a blocking cycle for $A$. This proves that $B$ cannot contain any cycles.
    Let us now consider a path $P=(x_1,\dots,x_p)$ in $B$ that is  maximal (in the sense that no contracts can be added to $P$ to obtain a longer path). By the acyclicity of $B$, $s:=s(x_1)$ must be a firm with no upstream contracts in $B$, and $t:=b(x_p)$ must be a firm with no downstream contracts in $B$. 
    
    Since $B$ is blocking for $A$, we know that $B$ is $(A,s)$-acceptable. Recall that $s$ has a flow-based choice function. Therefore either $s$ is a terminal firm (always accepting every offered contract), or $s$ must obey the Kirchhoff inequality, and therefore can accept the downstream contract $x_1 \in X^S_s \cap B$ only if there is a less preferred downstream contract in $A$. Note that this means that the contract $x_1$ is \ch{$(A,s(x_1))$-acceptable}. Using the fact that $B$ is $(A,t)$-acceptable, we can argue analogously to show that the contract $x_p$ is \ch{$(A,b(x_p))$-acceptable}. 
    Finally, consider any intermediary firm $f$ lying on path $P$; assume that $f=b(x_i)=s(x_{i+1})$ for some $i \in \{1, \dots, p-1\}$. Again, either $f$ is terminal (and thus accepts all offers) or it is non-terminal and obeys the Kirchhoff equality. In the latter case, the path~$P$ is $(A,f)$-acceptable because it contains exactly one upstream contract and one downstream contract for $f$, namely $x_i$ and $x_{i+1}$. 
    Hence we get that $P$ is a blocking path for $A$, a contradiction.    
    \end{proof}
    
    Theorems~\ref{thm:path+cycle} and~\ref{thm:set-stable-flow} imply that deciding whether a stable outcome exists in a flow network is $\mathsf{NP}$-complete. 
    
    \begin{corollary}
    \label{cor:set-stable-flow}
    %Suppose that in a flow network choice functions are flow-based. Then 
    It is $\mathsf{NP}$-complete to decide if a flow network admits a stable outcome. 
    \end{corollary}
    
    \subsubsection{Stability in trading networks}
    \label{sec.tradestab}
    We finally turn to the existence of stable outcomes in trading networks which was first considered by \citet{HaKoNiOsWe:14}. They showed that (under certain conditions) stable outcomes are equivalent to outcomes that are immune to blocks along trails in which every firm can simultaneously offer and accept all its contracts in the trail. However, it is easy to see that in general trading networks path-or-cycle stability does not imply stability. But since flow networks are a special case of the trading networks that we consider in our general model (because flow-based choice functions are fully substitutable and satisfy IRC), Corollary~\ref{cor:set-stable-flow} implies that deciding whether a stable outcome exists is $\mathsf{NP}$-hard also in our trading networks model. 
    
    \begin{corollary}
    \label{cor:set-stable-general}
    Suppose that in a trading network choice functions satisfy full substitutability and IRC. 
    Then it is $\mathsf{NP}$-hard to decide if the trading network admits a stable outcome. 
    \end{corollary}
    
    In fact, dealing with stable outcomes is even trickier in trading networks. Our last result demonstrates that even \textit{verifying} whether an outcome is stable is computationally intractable in general trading networks.\footnote{The proof of Theorem~\ref{thm:path+cycle} shows that in flow networks verifying whether an outcome is stable can be done in time linear in the number of contracts.}

    Let $\textsc{Instability}$ be the following decision problem. An instance of $\textsc{Instability}$ is a 
    trading network with set of contracts $X$ and a set of agents $F$ (with choice functions that satisfy full substitutability and IRC) and an outcome $A \subseteq X$. The answer for an instance of $\textsc{Instability}$ is YES if the particular outcome $A$ is \emph{not} stable (that is, if there is a set of contracts $Z$ that blocks $A$),
    otherwise the answer is NO. 
    %DEFINE ORACLE; DEFINE NP-COMPLETE.  
    \begin{theorem}\label{grouphopeless}
    The $\textsc{Instability}$ problem is $\mathsf{NP}$-complete. 
    %\footnote{A problem is co-$\mathsf{NP}$-complete if its complement is $\mathsf{NP}$-complete.} 
    Moreover, if choice functions are represented by oracles, 
    then finding the right answer for an instance of $\textsc{Instability}$ might require an
    exponential number of oracle calls.
    \end{theorem}
    
    \begin{proof}[Proof of Theorem \ref{grouphopeless}]
    The $\textsc{Instability}$ problem clearly belongs to complexity class $\mathsf{NP}$, as verifying that a given set~$Z$ of contracts is a blocking set for an outcome~$A$ requires polynomial time.
    
    To show that $\textsc{Instability}$ is $\mathsf{NP}$-hard we reduce the $\mathsf{NP}$-complete \textsc{Partition} problem to
    $\textsc{Instability}$. An instance of the \textsc{Partition} problem is given by a $k$-tuple
    $\Pi=(a_1,a_2,\ldots,a_k)$ of positive integers such that $a_1\le a_2\le
    \ldots\le a_k$ holds. The answer to this problem is
    YES if and only if there is a subset $I$ of $\{1,2,\ldots,k\}$ such that
    $\sum_{i\in I}a_i=s$ where $2s=\sum_{i=1}^ka_i$. Given our instance~$\Pi$ of \textsc{Partition}, 
    we define an instance of \textsc{Instability} as follows. First, construct a trading network
    with firms $f$ and  $g$ and with contracts $y$ and $x_i$ such that $f=s(y)=b(x_i)$ and $g=b(y)=s(x_i)$ for
    $i\in\{1,2,\ldots,k\}$. Next, define choice function $C_\Pi^f$ with the help of
    $s=\frac 12\sum_{i=1}^ka_i$ by 
    \[
    C_\Pi^f(X|Y)=\left\{\begin{array}{rcl}
    (X|Y)&\qquad\qquad&\mbox{if }\sum\{a_i:x_i\in X\}\ge s,\\ 
    (X|\emptyset)&&\mbox{if }\sum\{a_i:x_i\in X\}< s.
    \end{array}
    \right.
    \]
    
    It is straightforward to check that $C_\Pi^f$ satisfies IRC, so let us check whether it is fully satisfiable as well. 
    First, since $f$ always accepts all upstream contracts, both SSS and CSC clearly hold for $f$ as a buyer (that is, conditions 1(a) and 2(a) are satisfied in Definition~\ref{def:full-sat}). 
    To check same-side substitutability as a seller for $f$, let us consider a fixed set $X$ of upstream contracts offered for $f$. Then $f$ either accepts all downstream contracts (in case $X$ is such that $\sum\{a_i:x_i\in X\}\ge s$) or it rejects every downstream contract (otherwise); in either case, the set of contracts rejected by $f$ as a seller satisfies SSS (that is, condition 1(b) in Definition~\ref{def:full-sat}). To check cross-side complementarity as a seller for $f$ (condition 2(b) in Definition~\ref{def:full-sat}), note that if $f$ rejects the contract $y$, then it must be the case that $\sum\{a_i:x_i\in X\} < s$ for the offered set $X$ of upstream contracts. But then $f$ will reject $y$ for any subset $X' \subseteq X$ too, so CSC follows as well. Hence, $C_\Pi^f$ is fully substitutable.
    
    Next, define $C_\Pi^g$ as follows:
    \small
    \[
    C_\Pi^g(Y|X)=\left\{\begin{array}{rl}
    (\emptyset|\emptyset)&\mbox{if } Y=\emptyset,\\
    (Y|X)&\mbox{if }Y=\{y\}\mbox{ and }\sum\{a_i:x_i\in X\}\le s,\\
        \! \! (Y|X\cap\{x_1,x_2,\ldots,x_t\})&\mbox{if }Y=\{y\}\mbox{ and $t$ is such that}\\
        &  \quad \sum\{a_i:x_i\in X, i\le t\}\le s< \sum\{a_i:x_i\in X, i \leq t+1\}.
    \end{array}
    \right. 
    \]
    \normalsize
    
    One can readily check that $C_\Pi^g$ also satisfies IRC. 
    To see that it is also fully substitutable, first notice that $g$ never rejects any upstream contracts, so it satisfies both SSS and CSC with $g$ as a buyer. To check the requirements for SSS with $g$ taking the role of a seller, let us fix a set $Y$ of upstream contracts. If $Y=\emptyset$, then $g$ rejects all downstream contracts. Otherwise (that is, if $Y=\{y\}$), suppose that $g$ rejects some $x_j$ from a set $X$ of offered downstream contracts. This means that there exists an index $t<j$ such that $\sum \{a_i : x_i \in X, i \leq t\} \leq s < \sum\{a_i:x_i\in X, i<t+1\}$. But then, for any superset $X' \supseteq X$ of dowstream contracts offered to $g$, the same condition will hold for some $t' \leq t<j$, and thus $x_j$ will again be rejected. This proves SSS with $g$ being a seller. To verify that $C_\Pi^g$ is also cross-side complementary with $g$ as a seller, it suffices to observe that any downstream contract rejected while $Y=\{y\}$ is offered to $g$ will get rejected again when $Y'=\emptyset$ is offered to $g$. Hence, we get that $C_\Pi^g$ is fully substitutable.
  
    So far, based on our instance $\Pi$ of \textsc{Partition}, we have
    determined a trading network. To finish the construction of our $\textsc{Instability}$ instance, we set an outcome
    $A=\emptyset$. We have to show that the answer to our instance of the \textsc{Partition} problem is
    YES if and only if $A=\emptyset$ is not stable.
    
    Assume now that the answer to our instance $\Pi$ of \textsc{Partition}
     is YES, that is, $\sum_{i\in I}a_i=s$ for some $I \subseteq \{1,2, \dots, k\}$. Define
    $X_I:=\{x_i:i\in I\}$ and $Y=\{y\}$. By the above definitions,
    $C_\Pi^f(X|Y)=(X|Y)$ and $C_\Pi^g(Y|X)=(Y|X)$, hence $X\cup Y$ blocks
    $A=\emptyset$, so $A$ is not stable.
    
    Assume now that $A=\emptyset$ is not stable. This means that there is
    a blocking set $Z$ to $A$ and define $I=\{i:x_i\in Z\}$, $X_I:=\{x_i:x_i\in
    Z\}$
    %Z\cap\{x_1,x_2,\ldots,x_k\}$,
    and $Y:=Z\cap \{y\}$. As $Z$ is blocking, we have $C_\Pi^f(X_I|Y)=(X_I|Y)$ and
    $C_\Pi^g(Y|X_I)=(Y|X_I)$. If $Y=\emptyset$ then
    $(Y|X_I)=C_\Pi^g(Y|X_I)=C_\Pi^g(\emptyset|X_I)=(\emptyset,\emptyset)$, so
    $Z=X_I\cup Y=\emptyset\cup\emptyset=\emptyset$, and hence $Z$ is not
    blocking. Otherwise, $Y=\{y\}$, and from $C_\Pi^g(Y|X_I)=(Y|X_I)$ we get that
    $\sum_{i\in I}a_i\le s$. Moreover, from $y\in C_\Pi^f(X_I,Y)$ we get that
    $\sum_{i\in I}a_i\ge s$. Consequently $\sum_{i\in I}a_i= s$, and the answer
    to the \textsc{Partition} problem is YES.

    To prove the second part of the theorem, define a trading network
    with firms $f$ and  $g$ and with contracts $y$ and $x_i$ such that $f=s(y)=b(x_i)$ and $g=b(y)=s(x_i)$ for for $1\le i\le 2n$. Define the following choice function 
    \begin{equation}\label{cf}
    C_0^f(X|Y)=\left\{\begin{array}{rcl}
    (X|Y)&\qquad\qquad&\mbox{if }|X|\ge n+1\\ 
    (X|\emptyset)&&\mbox{if }|X|\le n
    \end{array}
    \right.
    \end{equation}
    For $I\subseteq \{1,2,\ldots,n\}$ define $X_I:=\{x_i:i\in I\}$. For $|I|=n$ let
    \[C_I^f(X|Y)=\left\{\begin{array}{rcl}
    (X|Y)&\qquad\qquad&\mbox{if }|X|\ge n+1\mbox{ or if }X=X_I\\ 
    (X|\emptyset)&&\mbox{if }|X|\le n\mbox{ and }X\ne X_I
    \end{array}
    \right.
    \]
    It is straightforward to check that choice functions $C^f_0$ and $C^f_I$
    above satisfy the full substitutability condition and IRC.
    Define the following choice function for $g$
    \begin{equation}\label{cg}
    C^g(Y|X)=\left\{\begin{array}{rcl}
    (\emptyset|\emptyset)&\qquad&\mbox{if } Y=\emptyset\\
    (Y|X)&&\mbox{if }Y=\{y\}\mbox{ and }|X|\le n\\
    (Y|X\cap\{x_1,x_2,\ldots,x_t\})&&\mbox{if }Y=\{y\}\mbox{ and
    }|\{x_i\in X: i\le t\}|= n
    \end{array}
    \right.
    \end{equation}
    As $C^g=C^g_\Pi$ for $\Pi=(1,1,\ldots,1)$, $C^g$ also satisfies the full substitutability condition and IRC.
    
    Now assume that an instance of $\textsc{Instability}$ is given by the above network
    and an outcome $A=\emptyset$. Assume that the choice functions are not given
    explicitly, but by value-returning oracles. Moreover, we know exactly that
    the choice function of $g$ is the one defined in \eqref{cg} and we know
    that the choice function of $f$ is either $C^f_0$ or $C^f_I$ for some $I$. 
    It is easy to check that $A$ is not stable if and only if $C^f=C^f_I$ and
    in this case the only blocking set is $Z=X_i\cup\{y\}$. So if one has to
    decide stability of $A$, then one must determine the $C^f(Z)$ values
    for all such possible $Z$, and this means $\binom{2n}n$ oracle calls.
    \end{proof}
    
    \section{Open question}\label{sec.open}
    \citet*{FlJaTaTe:16} introduce another solution concept, called \textit{weak trail stability}. Let us consider a trail $T=\{x_{1},...,x_{M}\}$ whose elements are arranged in a sequence $(x_{1},...,x_{M})$ and define $T_{f}^{\le m}=\{x_{1},...,x_{m}\}\cap T_{f}$ to be firm $f$'s contracts out of the first $m$ contracts in the trail
and $T_{f}^{\ge m}=\{x_{m},...,x_{M}\}\cap T_{f}$ to be firm $f$'s contracts out of the last $M-m+1$
contracts in the trail (where $m \in \{1,\ldots,M\}$). 

\begin{definition}\label{trailstable} An outcome $A\subseteq X$ is \textit{weakly trail-stable} if 
\begin{enumerate}
\item $A$ is acceptable.
\item There is no trail $T=\{x_1,x_2,\ldots,x_M\}$, such that $T\cap
  A=\emptyset$  and %%either (a),(b1),(c) or (a),(b2),(c) holds:
\begin{enumerate}
\item $\{x_1\}$ is $(A,f_1)$-acceptable for $f_1=s(x_{1})$ and
\item At least one of the following two options holds:
\begin{enumerate}
\item   $T_{f_{m}} ^{\le m}$ is 
$(A,f_m)$-acceptable for $f_m=b(x_{m-1})=s(x_{m})$ whenever $1<m\le M$, or
\item $ T_{f_{m}}^{\ge m-1}$ is
$(A,f_m)$-acceptable for $f_m=b(x_{m-1})=s(x_{m})$ whenever $1<m\le M$
\end{enumerate}
\item $\{x_M\}$ is $(A,f_{M+1})$-acceptable for $f_{M+1}=b(x_{M})$.
\end{enumerate}
The above trail $T$ is called a \textit{sequentially blocking trail to $A$}.

\end{enumerate}
\end{definition}
    
    Like trail stability, weak trail stability also has an intuitive economic interpretation. But note that as the sequentially blocking trail grows, we ensure that each intermediate
    agent wants to choose \textit{all} his contracts along the sequentially blocking trail. We did not impose this requirement for trail stability. As \citet*{FlJaTaTe:16} argue weak trail stability might be a good solution concept in markets where contracts are not executed quickly. \citet*{FlJaTaTe:16} also show that weak trail stability is a weaker solution concept than trail stability in trading networks under full substitutability (though not in general). However, unlike trail-stable outcomes, the set of weakly trail-stable outcomes does not appear to have a structure that allows for efficient computation. 
    
    \begin{conjecture}
    Suppose that in a trading network choice functions satisfy full substitutability and IRC. Then it is $\mathsf{NP}$-complete to decide whether there exists a weakly trail-stable outcome that is not trail-stable.
    \end{conjecture}
    
    \section{Conclusion}\label{sec.conclusion}
    
    We showed that stable outcomes, which are immune to deviation by groups of agents and which can be efficiently computed in two-sided matching markets or in supply chains, cannot be efficiently verified or computed in trading networks. Our main result is particularly strong: even in flow networks, outcomes that are immune to blocks in which even two agents can coordinate on an upstream and downstream contract cannot be computed efficiently. We suggest an alternative solution concept---trail stability---which, under full substitutability, always exists, has an intuitive economic interpretation, coincides with competitive equilibrium in model with prices, and can be efficiently found in general trading networks. Further work can examine the computation properties of similar stability notions in trading networks where agents have more complex preferences \citep{jagadeesan2017complementary}.

    \newpage
    \singlespacing
    % Bibliography
    \bibliographystyle{chicago}
    \bibliography{bib}
    
    \end{document}